\documentclass[sn-mathphys]{sn-jnl}
\jyear{2024}%

\usepackage{natbib} 
\usepackage{bm}

\theoremstyle{thmstyleone}%

\theoremstyle{thmstyletwo}%

\theoremstyle{thmstylethree}%
\newtheorem{theorem}{Theorem}%
\newtheorem{definition}{Definition}%
\newtheorem{proposition}[theorem]{Proposition}%
\newtheorem{lemma}[theorem]{Lemma}%

\renewcommand{\x}{\bm{x}}
\renewcommand{\y}{\bm{y}}
\newcommand{\Xp}{\mathbf{X}_p}
\newcommand{\Xd}{\mathbf{X}_d}
\newcommand{\Xs}{\mathbf{X}_s}
\newcommand{\xp}{\bm{x}_p}
\newcommand{\xd}{\bm{x}_d}

\newcommand{\fp}{f_p}
\newcommand{\fd}{f_d}
\newcommand{\fs}{f_s}

\newcommand{\F}{\mathbf{F}}
\newcommand{\G}{\mathbf{G}}
\newcommand{\n}{\mathbf{n}}
\newcommand{\Ho}{H_\circ}

\normalbaroutside 

\begin{document}

\title{Bounding seed loss from isolated habitat patches}

\author*[1]{\fnm{Benjamin} \sur{Hafner}}\email{benrhafner@gmail.com}

\author*[1]{\fnm{Katherine} \sur{Meyer}}\email{kjmeyer@carleton.edu}

\affil[1]{\orgdiv{Math and Statistics Department}, \orgname{Carleton College}, \\ \orgaddress{\street{One North College Street}, \city{Northfield}, \postcode{55057}, \state{MN}, \country{USA}}}


\abstract{Dispersal of propagules (seeds, spores) from a geographically isolated population into an uninhabitable matrix can threaten population persistence if it prevents new growth from keeping pace with mortality. Quantifying propagule loss can thus inform restoration and conservation of vulnerable populations in fragmented landscapes. To model propagule loss in detail, one can integrate dispersal kernels (probabilistic descriptions of dispersal) and plant densities. However, one might lack the detailed spatial information and computational tools needed for such integral models. Here we derive two simple upper bounds on the probability of propagule loss---one assuming rotational symmetry of dispersal and the other not---that require only habitat area, habitat perimeter, and the mean dispersal distance of a propagule. We compare the bounds to simulations of integral models for the population of \textit{Asclepias syriaca} (common milkweed) at McKnight Prairie---a 13.7 hectare reserve surrounded by agricultural fields in Goodhue County, Minnesota---and identify conditions under which the bounds closely estimate propagule loss.
}


\maketitle

\vspace{-0.5cm}

\bmhead{Acknowledgments}
We are grateful for feedback on a draft of this manuscript from the Shaw Lab and associates at UMN, including Allison Shaw, Dennis Kim, Martha Torstenson, Naven Narayanan, Jessica Valiarovski, and Matt Michalska-Smith. Carleton College supported KM's work on the project during a sabbatical.

\section{Introduction}\label{sec1}

This paper concerns the probability that propagules (seeds, spores) from a population in an isolated habitat patch disperse outside of that patch. 
In the right conditions, dispersal may help individuals reach environments with lower conspecific competition, fewer natural enemies, and new opportunities for establishment \citep{beckman2023causes,levin2003ecology}. 
Dispersal can also support metapopulation persistence despite local extinctions by allowing a species to (re)colonize patches and bet-hedge against fluctuating habitat quality
\citep{levin2003ecology,hanski2001population}.
But for a small site surrounded by a large matrix of unviable
habitat, dispersal out of the patch threatens local population persistence when emigration outweighs the arrival of propagules from external sources  \citep{hanski2001population}.
Cheptou and colleagues \citeyearpar{cheptou2008rapid} documented a striking urban example of this dispersal risk in the weed \textit{Crepis sancta}, whose dispersing seeds departed from tree plantings an estimated 55\% of the time, resulting in selection for a nondispersing seed variant. 

The proportion of propagules that leave viable habitat emerged as a key parameter for predicting population crashes in the colonization-mortality model of Cooney et al. \citeyearpar{cooney2022effect}---a reformulation of Tilman's \citeyear{tilman1994habitat} spatially implicit extinction debt model  that includes spatial dispersal processes. 
Quantifying propagule movement from inside to outside of habitat patches 
could thus help managers to identify either
isolated remnant populations that are at risk of decline 
or
sites that are suitable to support restored populations. 

Propagule loss from a habitat patch can be modeled in detail by integrating functions that represent the density of propagule sources and the probabilistic patterns of propagule dispersal (the dispersal kernel of the propagule source) \citep{cooney2022effect,klausmeier1998extinction}. In particular, the density of propagules arriving at a location $(x_0,y_0)$ in $\mathbb{R}^2$ is given by \citep{nathan2012dispersal}
\begin{equation}\label{eq:nathan}
    \substack{\text{propagule arrival}\\ \text{density at } (x_0,y_0)}=\iint\limits_{\text{habitat}} \left(\substack{\text{propagule} \\ \text{source density}\\ \text{at }(x,y)} \right)\left(\substack{\text{dispersal kernel}\\ \text{evaluated at} \\ (x_0-x,y_0-y)} \right) \, dA
\end{equation}
and total propagule loss can be obtained by integrating propagule arrival density over the complement of the habitat patch:
\begin{equation}\label{eq:seedloss}
    \text{propagule loss }=\iint\limits_{\substack{\text{not} \\ \text{habitat}}} \left(\substack{\text{propagule}\\ \text{arrival density}}\right) \ dA.
\end{equation}
This approach incorporates information about the spatial distribution of both propagule sources and propagule dispersal. In many  studies, knowledge of these attributes may  be incomplete.
\begin{figure}[h!]
    \centering
    \includegraphics[width=0.25\textwidth]{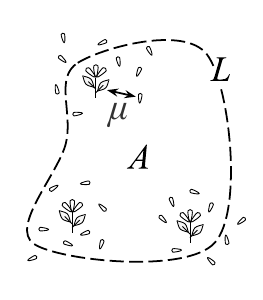}
    \vspace{1em}
    \caption{Schematic of quantities used to bound seed loss, with the habitat boundary depicted as a dashed line}
    \label{fig:intro}
\end{figure}

Here we derive an estimate of propagule loss that could be especially useful when one has only rough information on the location of propagule sources and the patterns of propagule dispersal. For brevity, we refer to propagule sources as plants and propagules as seeds. We make the simplifying assumptions that 
\begin{itemize}
    \item[(i)] plants are uniformly distributed within a viable habitat and
    \item[(ii)] the dispersal kernel associated with a plant is independent from its location in the habitat.
\end{itemize} 
In this context, we can bound the probability of seed loss $p$ from a population in an isolated habitat patch by
\begin{equation}\label{eq:main1}
p\leq \frac{\mu L}{2A}
\end{equation}
where $A$ is the habitat area, $L$ is the perimeter, and $\mu$ is the mean seed dispersal distance (illustrated in Figure \ref{fig:intro}).
If in addition we assume that dispersal patterns are identical in each direction---that is, rotationally symmetric---then we can decrease this upper bound to
\begin{equation}\label{eq:main2}
    p\leq \frac{\mu L}{\pi A}.
\end{equation}
Given that patch area $A$ appears in the denominator of these expressions, while patch perimeter $L$ appears in the numerator, we expect the seed loss bound to decrease as we scale up a habitat region while maintaining its shape (since $L$ grows linearly while $A$ grows quadratically). Loosely speaking, the ratio $A/L$ can be interpreted as the habitat's characteristic length scale, which must be kept larger than $\mu$ in order to limit seed loss. 

In addition, we observe that these bounds are tight, in two senses. In a practical sense, when applied to an ecologically relevant plant population, the bounds are close to the true value of $p$, especially in the rotationally symmetric case when the habitat is large relative to the mean dispersal distance (see Section \ref{sec:examples}). And from a mathematical perspective, they are tight in the sense that the constant factors $1/2$ and $1/\pi$ cannot be improved (see Appendix A).

Though simple to state, inequalities \eqref{eq:main1} and \eqref{eq:main2} require some creative vector calculus to prove. We begin by establishing a probabilistic modeling framework to precisely define $p$ in Section \ref{sec:modeling}. Then, in Section \ref{sec:vf}, we translate the spatial geometry of seed movement from dispersal kernels to vector fields and use the divergence theorem to reframe seed loss as a flux of seeds across the habitat boundary. In Section \ref{sec:bound} we bound seed flux across the habitat boundary. We put all the pieces together in Theorem \ref{thm:main}, which formalizes the claims \eqref{eq:main1} and \eqref{eq:main2}.
Once the proof of Theorem \ref{thm:main} is complete, we provide an example of how it could be applied in a conservation context in Section \ref{sec:examples}. Specifically, we compute seed loss bounds for a population of \textit{Ascelpias syriaca} (common milkweed) located at McKnight Prairie, an isolated reserve surrounded by agriculture, and show it closely matches an integral model of seed loss. In Section \ref{sec:disc} we discuss methods for estimating mean dispersal distance as well as future theoretical directions.

\section{Modeling framework}\label{sec:modeling} 

Throughout this text, given sets $U$ and $V\subset\mathbb{R}^2$ we use $\partial U$ to denote the boundary of $U$, $U^c$ for set complement, and $U\backslash V$ for the set difference $U\cap V^c$.

We represent viable habitat for a plant population of interest as an open subset $H\subset\mathbb{R}^2$
with finite area $A$ and perimeter $L$. 
We treat plant positions, seed displacements, and seed landing locations as continuous random vectors with values in $\mathbb{R}^2$.

The random vector $\Xp$ models the position of plants in the population. 
Recall that $\Xp$ relates to its probability density function $\fp$ as follows: given a set $U\subset\mathbb{R}^2$,
\begin{equation}
P(\Xp\in U)=\iint\limits_{\x\in U} \fp(\x) \, dA.
\end{equation}
Here $dA$ refers to an infinitesimal area of integration (e.g. $dx_1 dx_2$) and does not imply a relationship to total habitat area $A$. We assume plant positions are uniformly distributed over $H$, so 
\begin{equation}\label{eq:fp}
\fp(\x)=\begin{cases}1/ A & \text{if } \x\in H \\
0 & \text{otherwise}
\end{cases}
\end{equation}

The seeds released by a plant are displaced by a random dispersal vector $\Xd$ from the source plant before they land. The probability density function $\fd$ associated with $\Xd$ is also known as a dispersal kernel \citep{nathan2012dispersal}. Importantly, we assume $\Xp$ and $\Xd$ are independent: the position of a source plant within the habitat does not influence the dispersal pattern of its seeds (see Section \ref{sec:disc} for discussion of this assumption). The mean dispersal distance of seeds is given by
\begin{equation}\label{eq:mu}
\mu=E[ \, |\Xd| \, ]=\iint\limits_{\x\in \mathbb{R}^2} |\x|\fd(\x) \, dA.
\end{equation}
We assume that $\fd$
decays rapidly enough as $|\x|\to\infty$ to make $\mu$ well-defined. 
Commonly used functional forms of $\fd$ include the inverse Gaussian for wind-dispersed species and the exponential power function \citep{bullock2017synthesis,nathan2012dispersal}.

\begin{figure}[h]
    \centering
    \includegraphics[width=\textwidth]{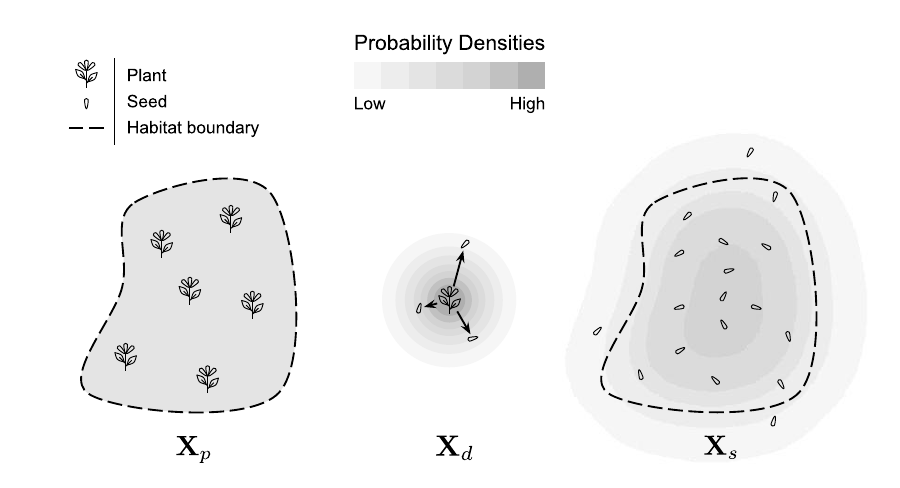}
    \caption{
    Visualization of the probability densities associated with random variables for plant locations $\Xp$, seed dispersal $\Xd$, and seed landing locations $\Xs$ 
    }
    \label{fig:PDFs}
\end{figure}

We define the random vector $\Xs$ by
\begin{equation}
    \Xs = \Xp + \Xd
\end{equation}
to model the landing position of a seed that originated from the plant population in $H$. 
We assume that $H^c$ is not viable habitat for the population of interest, so seeds that originate in $H$ but land in $H^c$ are lost. Our main quantity of interest, the probability $p$ of seed loss from the habitat, is given by
\begin{equation}\label{eq:pdef}
    p=P(\Xs\in H^c)=\iint\limits_{\x\in H^c}\fs(\x) \, dA
\end{equation}
where $\fs$ is the probability density of $\Xs$.

As a first step towards bounding seed loss probability $p$, we unpack $p$'s defining integral from (\ref{eq:pdef}) in terms of habitat area $A$ and dispersal kernel $\fd$.
\begin{proposition}\label{prop:pHfd}
The seed loss probability $p$ defined in equation (\ref{eq:pdef}) is equivalent to
\begin{equation}\label{eq:pHfd}
p=\frac{1}{A}\iint\limits_{\xp \in H} \ \ \iint\limits_{\x \in H^c} \fd(\x - \xp) \, dA \, dA.
\end{equation}
\end{proposition}
\begin{proof}
Because $\Xp$ and $\Xd$ are independent, the density of their sum $\Xs$ is the convolution of their densities:
\begin{equation}\label{eq:fs0}    \fs(\x)=\iint\limits_{\xp\in\mathbb{R}^2}\fp(\xp)\fd(\x-\xp) \, dA.
\end{equation}
Replacing $\fp$ in (\ref{eq:fs0}) with its definition from (\ref{eq:fp}) yields
\begin{equation}\label{eq:fs1}
    \fs(\x)=\frac{1}{A}\iint\limits_{\xp\in H}\fd(\x-\xp) \, dA.
\end{equation}
By substituting $\fs$ from (\ref{eq:fs1}) into the definition (\ref{eq:pdef}) of $p$ and reversing the order of integration, one obtains the expression \eqref{eq:pHfd}. 
\end{proof}

The representation of $p$ in  Proposition \ref{prop:pHfd} can be interpreted as seed loss probability for an individual plant at $\xp$ (the inner integral), averaged uniformly over the habitat $H$ (the outer integral). This form is useful for calculating $p$ numerically from a known habitat $H$ and dispersal kernel $\fd$, as done by \cite{cooney2022effect}.

\section{Seed loss as the flux of a vector field across habitat boundary}\label{sec:vf}

To bound $p$ in terms of habitat perimeter $L$, we  transform the inner integral over $H^c$ in equation \eqref{eq:pHfd} into a line integral around the habitat boundary $\partial H$. Specifically, we  express seed loss as a flux of seeds crossing the habitat boundary. To do so, we must describe seed flow in the language of vector fields. Subsection \ref{sec:vfsingle} constructs a vector field describing dispersal from a single plant and subsection \ref{sec:vftotal} integrates over all plant locations to obtain a vector field representing total seed dispersal.

\subsection{Vector field of a single plant}\label{sec:vfsingle}

\begin{figure}[b!] \label{fig:F}
    \centering
    \includegraphics[width=0.3\textwidth]{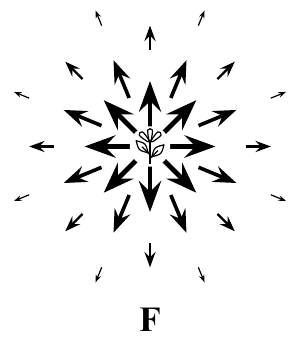}
    \vspace{1em}
    \caption{
    The vector field $\F$ represents the spatial flow of seeds away from a source plant 
    }
\end{figure}

The dispersal of seeds from a single plant is typically represented as a probability density function $f_d$, but the same information can also be represented as a vector field, which we call $\F$. The vector field $\F$ describes the motion of seeds as they disperse---like a flowing fluid, whereas $f_d$ describes their final resting positions.

The mathematical link between $\F$ and $f_d$ is the divergence operator. Intuitively, a positive divergence of $\F$ indicates a source (seeds being released) and a negative divergence indicates a sink (seeds landing). Our definition of $\F$ below yields a positive divergence singularity at the origin, representing the plant as a point source of seeds, and yields $\nabla\cdot\F = -f_d$ everywhere else, encoding seed landing positions. For more physical interpretation and motivation of $\F$'s definition, see Appendix \ref{secAF}.

\begin{definition}\label{def:F} Given a dispersal kernel $\fd$, the corresponding vector field $\F$ is
\begin{equation}\label{eq:F}
    \F(\x)=\frac{\x}{r_x^2}\int_{r_x}^\infty \fd(r,\theta_x) \, r \, dr.
\end{equation}
\end{definition}

\begin{lemma}\label{lem:Ffd}
For all non-zero $\x\in\mathbb{R}^2$,
\begin{equation}
    \nabla\cdot\F(\x)=-\fd(\x).
\end{equation}
\end{lemma}

\begin{proof}
    We work in polar coordinates with unit vectors $\bm{\hat r}=\x/|\x|$ and $\bm{\hat\theta}$ the rotation of $\bm{\hat r}$ by $\pi/2$.
    In these coordinates, $\F(\x)=|\F(\x)|\bm{\hat r}+0\,\bm{\hat \theta}$ and by the polar divergence formula
    \begin{align}
        \nabla\cdot\F&=\frac{1}{r_x}\frac{\partial}{\partial r_x}\Big[r_x |\F| \Big]+\frac{1}{r_x}\frac{\partial}{\partial\theta}\Big[0\Big]\\
        &=\frac{1}{r_x}\frac{\partial}{\partial r_x}\Big[\int_{r_x}^\infty \fd(r,\theta_x) \, r \, dr\Big]\\
        &=\frac{1}{r_x}\big(-\fd(r_x,\theta_x)r_x\big)\\
        &=-\fd(r_x,\theta_x).
    \end{align}
\end{proof}

Using Definition \ref{def:F} and Lemma \ref{lem:Ffd}, one can convert from a dispersal kernel $\fd$ to a vector field $\F$ and back again. The two objects carry the same information about seed dispersal, simply in different forms. 

The remainder of Section \ref{sec:vfsingle} concerns the flux of $\F$ across the habitat boundary, which represents a single plant's probability of seed loss. 

\begin{lemma}\label{lem:Fcircle}
    Let $D$ be a disk shaped habitat with a single plant at its center, the origin. Then the flux of $\F$ across the habitat boundary $\partial D$ is equal to the probability that the plant's seeds are lost (land outside $D$). That is,
    \begin{equation}
        \oint\limits_{\x \in \partial D} \F(\x)\cdot\n \, ds = \iint\limits_{\x \in D^c} \fd(\x) \, dA
    \end{equation}
    where $\n$ is the outward pointing unit vector normal to the boundary.
\end{lemma}
\begin{proof}
    Suppose the disk has radius $R$. Then applying Definition \ref{def:F} and simplifying yields
    \begin{align}
        \oint\limits_{\x \in \partial D}\F(\x)\cdot\n \, ds&=\oint\limits_{\x \in \partial D}  
        \frac{1}{R}\int_{R}^\infty \fd(r,\theta_x) \, r \, dr \, ds.
    \end{align}
    Replacing $ds$ with $R \, d\theta_x$ and parameterizing the loop integral by $\theta_x$ gives
    \begin{align}
        \oint\limits_{\x \in \partial D}\F(\x)\cdot\n \, ds&=\int_{0}^{2\pi}  
        \frac{1}{R}\int_{R}^\infty \fd(r,\theta_x) \, r \, dr \, R \, d\theta_x\\
        &=\int_{0}^{2\pi}  
        \int_{R}^\infty \fd(r,\theta_x) \, r \, dr \, d\theta_x\\
        &=\iint\limits_{\x \in D^c} \fd(\x) \, dA.
    \end{align}
\end{proof}

In Proposition \ref{prop:div}, we use the divergence theorem to generalize Lemma \ref{lem:Fcircle} from a disk $D$ to a generic habitat shape $H$ and allow the plant to reside at any point $\xp$ in $H$, not just the origin.

\begin{proposition}\label{prop:div} 
For all $\xp\in H$,
\begin{equation}\label{eq:div}
\iint\limits_{\x \in H^c} \fd(\x - \xp) \, dA=\oint\limits_{\x \in \partial H} \F(\x - \xp) \cdot \n \, ds.
\end{equation}
\end{proposition}

\begin{proof}
We begin with the right hand side and derive the left. To apply the divergence theorem, care is required around the singularity of $\F(\x-\xp)$ when $\x=\xp$. Fix $\xp\in H$ and choose $\epsilon>0$ such that the ball $B_\epsilon(\xp)$ of radius $\epsilon$ around $\xp$ lies entirely within $H$. Let $H_\circ=H\backslash B_\epsilon(\xp)$.
Note that the boundary of $\Ho$ consists of the boundary of $H$ and the boundary of $B_\epsilon(\xp)$. Let $\n$ represent the unit normal to each boundary, with orientations shown in Figure \ref{fig:orientation}.  We have that
\begin{equation}\label{eq:boundaries}
    \oint\limits_{\x\in\partial H}\F(\x-\xp)\cdot\n \, ds
    =\underbrace{\oint\limits_{\x\in\partial\Ho}\F(\x-\xp)\cdot\n \ ds }_{\text{I}}
    -\underbrace{\oint\limits_{\underset{\large\partial B_\epsilon(\xp)}{\x\in}} \F(\x-\xp)\cdot\n \, ds}_{\text{II}}.
\end{equation}

\begin{figure}[h]
    \centering
    \includegraphics[width=0.25\textwidth]{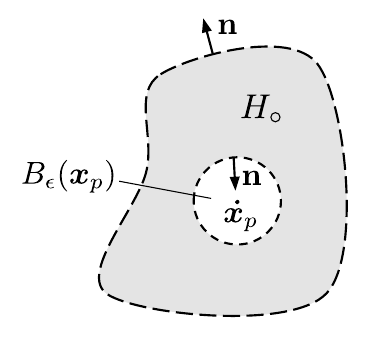}  
    \vspace{1em}
    \caption{
    Orientation of normal vectors used in the proof of Proposition \ref{prop:div}
    }
    \label{fig:orientation}
\end{figure}

\noindent Because $\Ho$ excludes the singularity at $\xp$, we may apply the divergence theorem to integral I of (\ref{eq:boundaries}). This gives
\begin{equation}\label{eq:intI}
\begin{aligned}
   \oint\limits_{\x\in\partial\Ho}\F(\x-\xp)\cdot\n \ ds&= \iint\limits_{\x\in\Ho} \nabla\cdot\F(\x-\xp) \, dA\\ &=   -\iint\limits_{\x\in\Ho} \fd(\x-\xp) \, dA,
\end{aligned}
\end{equation}
where the second equality follows from Lemma \ref{lem:Ffd}.
Integral II is handled by Lemma \ref{lem:Fcircle}, 
since $B_\epsilon(\xp)$ is a disk:
\begin{equation}\label{eq:intII}
    \oint\limits_{\x\in \partial B_\epsilon(\xp)} \F(\x-\xp)\cdot\n \, ds =
    -\iint\limits_{\x\in B_\epsilon(\xp)^c} \fd(\x-\xp) \, dA.
\end{equation}
Note the translation from $\x$ in Lemma \ref{lem:Fcircle} to $\x - \xp$ in equation \eqref{eq:intII}, along with the re-orientation of the unit normal from outward to inward, resulting in negation.


Substituting (\ref{eq:intI}) and (\ref{eq:intII}) for integrals I and II in (\ref{eq:boundaries}) gives
\begin{equation}\label{eq:two_ints}
    \oint\limits_{\x\in\partial H}\F(\x-\xp)\cdot\n \, ds
    = -\iint\limits_{\x\in\Ho} \fd(\x-\xp) \, dA+\iint\limits_{\x\in B_\epsilon(\xp)^c} \fd(\x-\xp) \, dA. 
\end{equation}
Because $\Ho\subset B_\epsilon(\xp)^c$, the difference of integrals in (\ref{eq:two_ints}) is the integral over their set difference $B_\epsilon(\xp)^c\backslash\Ho$, which is $H^c$.
This yields our desired result
\[
\oint\limits_{\x \in \partial H} \F(\x - \xp) \cdot \n \, ds=\iint\limits_{\x \in H^c} \fd(\x - \xp) \, dA.
\]
\end{proof}

\subsection{Integrating seed flux from all plant locations}\label{sec:vftotal}
Next we develop a second vector field $\G$ that aggregates seed movement from all plants in the habitat.
Using Proposition \ref{prop:div}, we can rewrite the inner integral in the expression \eqref{eq:pHfd} for seed loss probability $p$ in terms of seed flux from a single source at $\xp$ across the habitat boundary. This gives
\begin{equation}
    p=\frac{1}{A}\iint\limits_{\xp\in H} \ \oint\limits_{\x\in\partial H} \F(\x-\xp)\cdot \n \, ds \, dA.
\end{equation}
Reversing the order of integration and moving $\n$ out of the inner integral yields

\begin{align}
    p
    &=\frac{1}{A}\oint\limits_{\x\in\partial H} \ \iint\limits_{\xp\in H}  \F(\x-\xp) \, dA \, \cdot \n \, ds. \label{eq:G0}
\end{align}
The inner integral in \eqref{eq:G0} now represents the effective flow of seeds contributed from all plants in the habitat at a point $\x$ on the habitat boundary. We call this the total dispersal field and denote it
\begin{equation}\label{eq:G}
\G(\x)\equiv\displaystyle\iint\limits_{\xp\in H} \F(\x - \xp) \, dA,
\end{equation}
so that
\begin{equation}\label{eq:pG}
    p=\frac{1}{A}\oint\limits_{\x\in\partial H} \G(\x)\cdot \n \, ds.
\end{equation}

\begin{figure}[t!]
    \centering
    \includegraphics[width=0.4\textwidth]{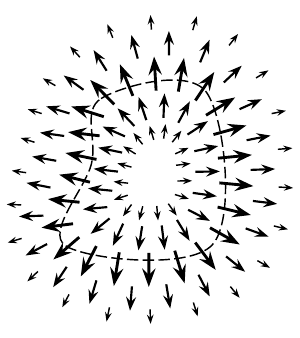}  
    \vspace{1em}
    \caption{
    The total dispersal field $\G$ models the net flow of all seeds originating from the plant population in $H$ (eq. \eqref{eq:G}). 
    The outward flux of $\G$ across the habitat boundary, depicted here as a dotted line, represents seed loss (eq. \eqref{eq:pG})
    }
    \label{fig:G}
\end{figure}

\section{Bounding seed loss}\label{sec:bound}

Based on equation \eqref{eq:pG}, in order to bound $p$ it suffices to bound the magnitude of $\G\cdot\n$. We derive two upper bounds on $\G\cdot\n$: one with no assumptions on the symmetry of the dispersal kernel $\fd$ (Lemma \ref{lem:Gmax_asym}) and another when $\fd$ is rotationally symmetric (Lemma \ref{lem:Gmax_symm}). 

\begin{lemma}
\label{lem:Gmax_asym}
For all $\x\in\mathbb{R}^2$ and for all unit vectors $\n\in\mathbb{R}^2$,
\[
\G(\x)\cdot\n \leq \frac{\n}{2}\cdot E(\Xd)+\frac{\mu}{2}.
\]
\end{lemma}
\begin{proof}
    It follows from the definition (\ref{eq:G}) of $\G$ that
\[
\G(\x)\cdot\n = \iint\limits_{\xp\in H} \F(\x-\xp)\cdot\n \, dA.
\]
To simplify the following argument, we define $\y=\x-\xp$ as a shorthand for the seed dispersal vector. Additionally, let $\hat{H}=\{\y\in\mathbb{R}^2 \, | \, \x-\y\in H\}$ so that 
\[
\G(\x)\cdot\n = \iint\limits_{\y\in\hat{H}} \F(\y)\cdot\n \, dA.
\]
Expanding $\F$ by Definition \ref{def:F}, we have
\begin{equation}
    \G(\x)\cdot\n = \iint\limits_{\y\in\hat{H}} \frac{\y\cdot\n}{r_y^2}\int_{r_y}^\infty \fd(r,\theta_y) \, r \, dr \, dA
\end{equation}
where $\theta_y, r_y$ are the polar coordinates of $\y$. Next, we use the fact that $\y\cdot\n \leq \frac{\y\cdot\n}{2} + \frac{|\y\cdot\n|}{2}$ to produce the inequality
\begin{equation}\label{eq:abs_val_trick}
    \G(\x)\cdot\n \leq \iint\limits_{\y\in\hat{H}} \left(\frac{\y\cdot\n}{2} + \frac{|\y\cdot\n|}{2}\right) \frac{1}{r_y^2}\int_{r_y}^\infty \fd(r,\theta_y) \, r \, dr \, dA.
\end{equation}
Substituting $\frac{\y\cdot\n}{2} + \frac{|\y\cdot\n|}{2}$ for $\y\cdot\n$ has no effect when $\y\cdot\n>0$, that is, when the seed dispersal vector points outward across the habitat boundary. But when $\y\cdot\n<0$, it is replaced by zero. Thus, we ignore the contribution of any seeds flowing inward across the habitat boundary. Although at first it appears to introduce needless complication, this substitution enables the rest of the proof.

Because $\frac{\y\cdot\n}{2} + \frac{|\y\cdot\n|}{2}$ is never negative, we are now free to expand the region of integration from $\y\in\hat{H}$ to all $\y\in\mathbb{R}^2$, which we parameterize by $\theta_y, r_y$, with $\theta_y$ measured relative to $\n$ so that $\y\cdot\n=r_y\cos\theta_y$. With this, \eqref{eq:abs_val_trick} becomes
\begin{align}
    \G(\x)\cdot\n
    &\leq\int_0^{2\pi}\int_0^\infty \left(\frac{r_y\cos\theta_y}{2}+\frac{|r_y\cos\theta_y|}{2}\right)\frac{1}{r_y^2}\int_{r_y}^\infty \fd(r,\theta_y) \, r \, dr \, r_y \, dr_y \, d\theta_y\\
    &=\int_0^{2\pi}\int_0^\infty \int_{r_y}^\infty \left(\frac{\cos\theta_y}{2}+\frac{|\cos\theta_y|}{2}\right) \fd(r,\theta_y) \, r \, dr \, dr_y \, d\theta_y.
\end{align}

We interchange the order of integration between $dr$ and $dr_y$ to obtain
\begin{equation}
   \G(\x)\cdot\n\leq \int_0^{2\pi} \int_0^\infty \int_0^r \left(\frac{\cos\theta_y}{2}+\frac{|\cos\theta_y|}{2}\right) \fd(r,\theta_y) \, r \, dr_y \, dr \, d\theta_y
\end{equation}
and evaluate the innermost integral, yielding an additional factor of $r$:
\begin{equation} 
   \G(\x)\cdot\n\leq\int_0^{2\pi}\ \int_0^\infty \left(\frac{r\cos\theta_y}{2}+\frac{r|\cos\theta_y|}{2}\right) \fd(r,\theta_y) \, r \, dr \, d\theta_y. \label{eq:Gdotn_pickup}
\end{equation}
Let $\hat\y=\frac{r}{r_y}\y$, so that $\hat\y$ has the same direction as $\y$ but magnitude $r$ instead of $r_y$. Then $r\cos\theta_y=\hat\y\cdot\n$ and 
\begin{align}
    \G(\x)\cdot\n&\leq\int_0^{2\pi}\int_0^\infty\left(\frac{\hat\y\cdot\n}{2}+\frac{|\hat\y\cdot\n|}{2}\right)\fd(r,\theta_y) \, r \, dr \, d\theta_y\\
    &=\iint\limits_{\hat\y\in\mathbb{R}^2} \left(\frac{\hat\y\cdot\n}{2}+\frac{|\hat\y\cdot\n|}{2}\right)\fd(\hat\y) \, dA\\
    &\leq\frac{\n}{2}\cdot\iint\limits_{\hat\y\in\mathbb{R}^2}\hat\y \, \fd(\hat\y) \, dA + \frac{1}{2}\iint\limits_{\hat\y\in\mathbb{R}^2}\big|\hat\y\big|\fd(\hat\y) \, dA\\
    &=\frac{\n}{2}\cdot E(\Xd)+\frac{1}{2}\mu
\end{align}
as claimed.
\end{proof}

The bound in Lemma \ref{lem:Gmax_asym} remains valid when seed dispersal is asymmetric or biased in one direction, for example by prevailing winds. In contrast, Lemma \ref{lem:Gmax_symm} requires that dispersal patterns are identical in each direction, which yields a stronger bound.

\begin{lemma} \label{lem:Gmax_symm}
Suppose $\fd$ is rotationally symmetric; that is, there exists a function $\hat\fd$ such that $\fd(r,\theta) = \hat\fd(r)$ for all $\theta$. Then for all $\x\in\mathbb{R}^2$ and for all unit vectors $\n\in\mathbb{R}^2$,
\[
\G(\x)\cdot\n \leq \frac{\mu}{\pi}.
\]
\end{lemma}
\begin{proof}
Following the proof of Lemma \ref{lem:Gmax_asym} to equation \eqref{eq:Gdotn_pickup}, we have
\begin{equation*}
    \G(\x)\cdot\n \leq \int_0^{2\pi}\ \int_0^\infty \left(\frac{r\cos\theta_y}{2}+\frac{r|\cos\theta_y|}{2}\right) \fd(r,\theta_y) \, r \, dr \, d\theta_y.
\end{equation*}
But because $\fd(r,\theta) = \hat\fd(r)$, which depends only on $r$, this double integral can now be separated:
\begin{align}
    \G(\x)\cdot\n &\leq \int_0^{2\pi} \left(\frac{\cos\theta_y}{2}+\frac{|\cos\theta_y|}{2}\right) d\theta_y \int_0^\infty \hat\fd(r) \, r^2 \, dr\\
    &=
    2 \int_0^\infty \hat\fd(r) \, r^2 \, dr. \label{eq:int_fdhat}
\end{align}
When the definition of $\mu$ from equation \eqref{eq:mu} is re-written in polar coordinates,
\begin{equation} \label{eq:mu_polar}
    \mu = \int_0^{2\pi} \int_0^\infty r \hat\fd(r) \, r \, dr \, d\theta = 2\pi \int_0^\infty \hat\fd(r) \, r^2 \, dr,
\end{equation}
it becomes clear, upon comparison of equations \eqref{eq:int_fdhat} and \eqref{eq:mu_polar}, that
\begin{equation}
    \G(\x)\cdot\n \leq \frac{\mu}{\pi}.
\end{equation}
\end{proof}

\begin{theorem}\label{thm:main}
Suppose a plant population with mean seed dispersal distance $\mu$ is uniformly distributed over a habitat patch $H$ with finite area $A$ and perimeter $L$. If dispersal is independent of plant position, then the probability that a given seed disperses outside the habitat is bounded by 
\[
p \leq \frac{\mu L}{2 A}.
\]
If in addition dispersal is rotationally symmetric, then the stronger bound
\[
p \leq \frac{\mu L}{\pi A}
\]
holds.
\end{theorem}
\begin{proof}
From equation \eqref{eq:pG} we have 
    \[
    p = \frac{1}{A} \oint\limits_{\x\in\partial H} \G(\x)\cdot\n(\x) \, ds,
    \]
where $\n(x)$ is the outward-oriented unit vector orthogonal to $\partial H$ at $\x$.
In the general case, Lemma \ref{lem:Gmax_asym} gives that
    \[
    \G(\x)\cdot\n(\x) \leq \frac{\n(\x)}{2}\cdot E(\Xd)+\frac{\mu}{2}, 
    \]
and so
\begin{align}
    p
    &\leq 
    \frac{1}{A} \oint\limits_{\x\in\partial H} \left( \frac{\n(\x)}{2}\cdot E(\Xd)+\frac{\mu}{2}\right) ds\\
    &=\frac{1}{A}\left(\frac{1}{2}E(\Xd) \, \cdot \oint\limits_{\x\in\partial H} \n(\x) \, ds + \oint\limits_{\x\in\partial H}\frac{\mu}{2} \, ds  \right).
\end{align}
The integral of $\n(\x)$ around the closed curve $\partial H$ vanishes, leaving 
\begin{equation}
    p
    \leq
    \frac{1}{A}
    \oint\limits_{\x\in\partial H}\frac{\mu}{2} \, ds
    =
    \frac{\mu L}{2A}
\end{equation}
as claimed.

In the special case that dispersal is rotationally symmetric, we have from Lemma \ref{lem:Gmax_symm} that 
    \[
    \G(\x)\cdot\n \leq \frac{\mu}{\pi} 
    \]
which immediately gives 
    \[
    p\leq \frac{1}{A} \oint\limits_{\x\in\partial H} \frac{\mu}{\pi} \, ds
    =\frac{\mu L}{\pi A}.
    \]
\end{proof}

In Appendix A we show that the constant factors $1/2$ and $1/\pi$ appearing in Theorem \ref{thm:main} are the best that can be achieved under the assumptions of this paper. 

\section{Example}\label{sec:examples}

To illustrate a practical application of Theorem \ref{thm:main}, we turn to McKnight Prairie, a 
14 hectare (34 acre) patch of remnant prairie in southern Minnesota. The site was mostly spared from agricultural use due to its hilly terrain, and is now protected under the Minnesota Department of Natural Resources (DNR) Native Prairie Bank Program.

Since the plot is embedded in a matrix of
commercial agriculture, any seeds that land outside its borders are effectively lost. Especially for a species like \textit{Ascelpias syriaca} (common milkweed), whose seeds are easily carried away on the wind, one might wonder if a large fraction of seeds are lost to the surrounding farmland, hampering the species' ability to reach available sites within the local prairie and potentially leading to population decline. In this type of scenario, Theorem \ref{thm:main} may 
provide a rough estimate and upper bounds on seed loss. We compute these bounds in Section \ref{sec:McK_bounds}, simulate seed loss via a full dispersal kernel model in Section \ref{sec:McK_sim}, and compare the bounds to simulations in Section \ref{subsec:compare}.

\subsection{Upper bounds on seed loss}\label{sec:McK_bounds}

To compute the bounds given by Theorem \ref{thm:main}, three inputs are required: habitat area $A$, habitat perimeter $L$, and mean dispersal distance $\mu$. We traced the border of McKnight Prairie from Google Maps and then used a simple Python program to estimate its area and perimeter: $A = 137,000$ m$^2$ (13.7 hectares) and $L = 1,930$ m. One could also make physical measurements or use GIS software.

Mean dispersal distance is somewhat more challenging to estimate. Section \ref{subsec:estimating_mu} discusses several practical routes to computing $\mu$, including the ``ballistic" model for wind-dispersed seeds,
\begin{equation}\label{eq:ballistic}
    \mu=\frac{v_wh}{v_T},
\end{equation}
which requires release height $h$, terminal velocity $v_T$, and wind speed $v_w$ \citep{nathan2012mechanistic}. For \textit{A. syriaca}, we take $h = 0.866$ m and $v_T = 0.219$ m/s from \cite{sullivan2018mechanistically}. We estimate wind speed using proxy data recorded by \cite{seeley2021hourly} at Cedar Creek Ecosystem Science Reserve, which is 100 km from McKnight Prairie. We estimate $v_w = 4.23$ m/s, which is the maximum hourly speed over September 17--23 (a seed-release period used by \cite{sullivan2018mechanistically}) averaged over all 19 years with complete wind data. Equation \eqref{eq:ballistic} then yields $\mu = 16.7$ m.

With these parameters ($A = 137,000$ m$^2$, $L = 1,930$ m, $\mu = 16.7$ m), Theorem \ref{thm:main} gives the following upper bounds on seed loss probability:
\begin{align*}
    p \leq \frac{\mu L}{\pi A} = 7.5\% && 
    p \leq \frac{\mu L}{2A} = 11.8\% \\
    \text{(Symmetric)}\quad && \text{(Asymmetric).}\quad
\end{align*}
These relatively low bounds are a result of the habitat's large size, on the order of 100's of meters, which gives a characteristic length scale $A/L$ that is much greater than $\mu$. Even in the asymmetric case, $p \leq 11.8\%$, meaning McKnight Prairie will retain almost 90\% of \textit{A. syriaca's} seeds. 
While additional factors such as mortality rates within the patch may interact with seed loss to determine the ultimate outcome, the species' high seed retention is encouraging.

\begin{figure}
    \centering
    \includegraphics[width=0.7\textwidth]{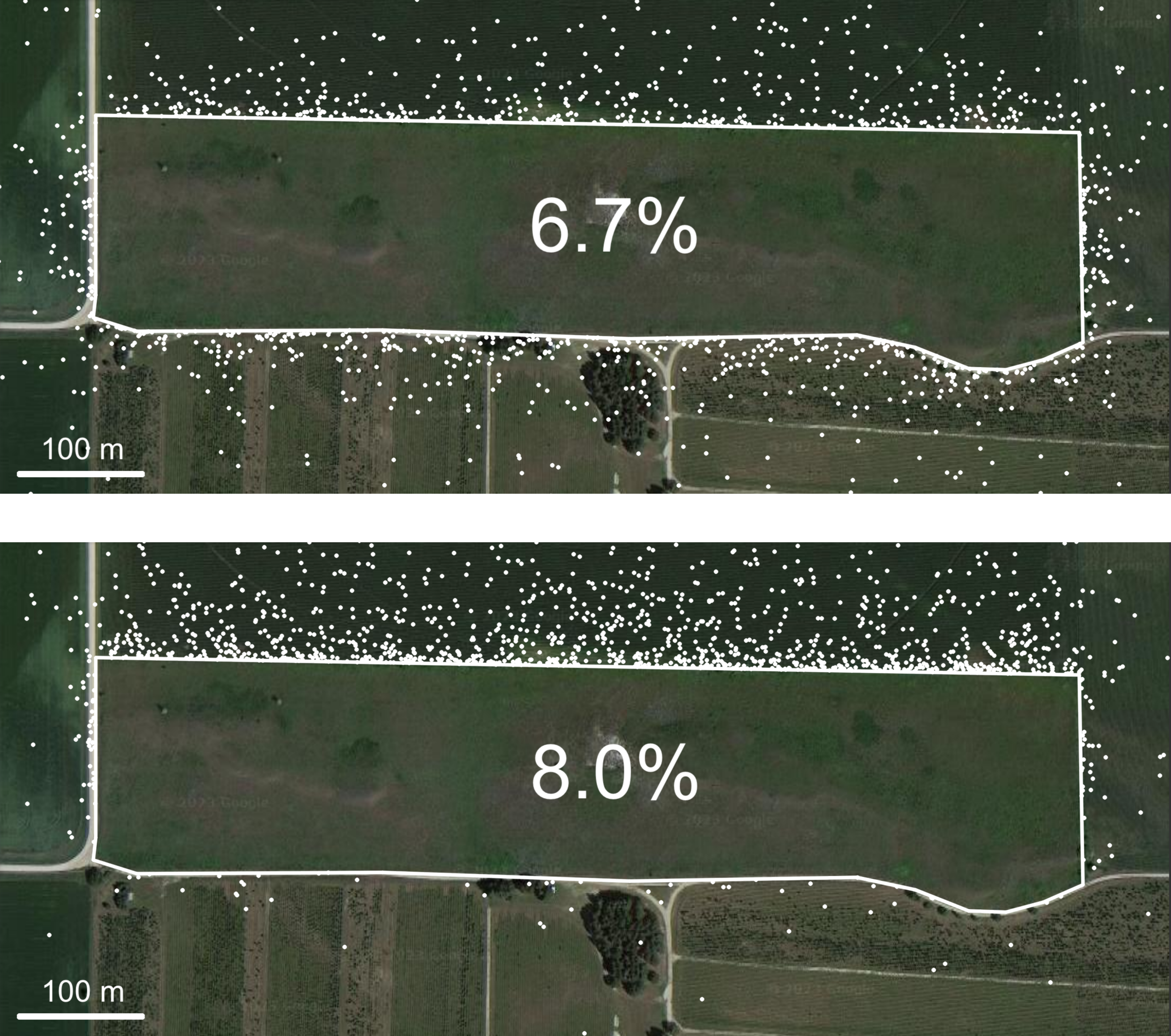}
    \caption{A random sample of several thousand \textit{A. syriaca} seeds lost from McKnight Prairie (outlined in white) to the surrounding agricultural fields. More seeds are lost when dispersal is biased northward (bottom panel) than when dispersal is symmetric (top panel) even though mean dispersal distance is the same in both cases. Aerial image: Google, \copyright2023 Airbus, CNES/Airbus, Maxar Technologies, USDA/FPAC/GEO}
    \label{fig:sim}
\end{figure}

\subsection{Simulation of seed loss}\label{sec:McK_sim}

In addition to using Theorem \ref{thm:main} to derive upper bounds on seed loss from \textit{A. syriaca} at McKnight prairie, we computed $p$ using a numerical simulation consistent with the probabilistic definition of $p$ in equation \ref{eq:pdef}. We implemented a simple Monte Carlo style simulation in which a plant is placed uniformly at random within the habitat and a seed disperses from there, landing either inside or outside the habitat. This process was repeated 10 million times, and the fraction of lost seeds was reported as $p$. We ran the simulation twice, once for symmetric dispersal and once for asymmetric dispersal. In both cases, dispersal distance was sampled from an inverse Gaussian (Wald) distribution with mean $\mu=16.7$ m (as computed above) and shape parameter $\lambda = 2.5$ m (as reported in \cite{sullivan2018mechanistically}). In the symmetric case, the direction of the seed dispersal vector was chosen uniformly at random, but in the asymmetric case, a northern bias was introduced. To simulate a hypothetical prevailing northern wind, the direction, $\theta$, measured relative to due north, was sampled according to the probability density function $(1 + 10\theta^2)^{-1}$ (normalized so that its integral from $-\pi$ to $\pi$ is 1). The results of the simulations were $p=6.7\%$ in the symmetric case and $p=8.0\%$ in the asymmetric case. As expected, these values are less than the theoretical bounds of 7.5\% and 11.8\%.

\begin{figure}
    \centering
    \includegraphics[width=\textwidth]{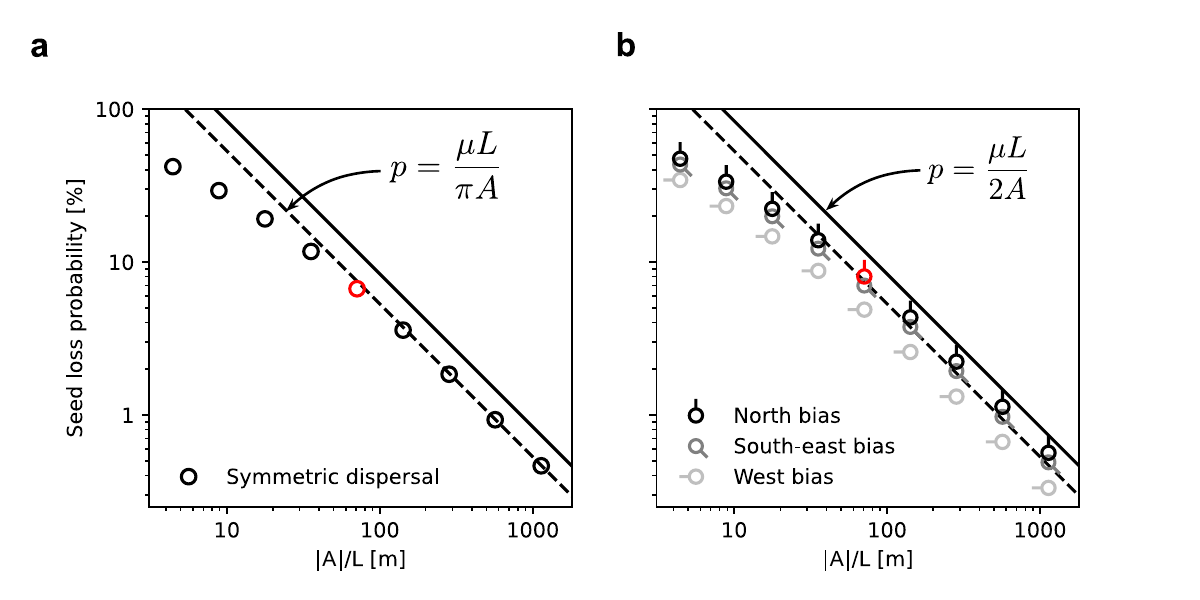}
    \caption{Impact of area to perimeter ratio $A/L$, on simulated seed loss probability $p$. Red symbols represent the two examples shown above, while other symbols represent hypothetical scenarios where McKnight Prairie is uniformly scaled up or down by a factor of 2, 4, 8, or 16, thus adjusting $A/L$. As the habitat grows larger, the symmetric bound (dashed line) becomes sharp, whereas the asymmetric bound (solid line) remains an over-estimate, especially for the alternative bias directions shown in gray
    }
    \label{fig:plots}
\end{figure}

\subsection{Comparison of seed loss bounds to simulations}\label{subsec:compare}

To illustrate a wider range of seed loss scenarios, we also simulated the effect of increasing or decreasing the habitat's area to perimeter ratio $A/L$, which is inversely proportional to the bounds on $p$. To change of $A/L$, we scaled the habitat size up or down uniformly. (Equivalently, we could have held the habitat size constant and scaled the mean dispersal distance $\mu$ down and up.)
As shown in Figure \ref{fig:plots}a, the simulated results trend closely with the theoretical bounds over a wide range of $A/L$ values. For larger $A/L$ values, the simulated seed loss probability for symmetric dispersal in particular shows remarkable agreement with the predicted bound. In these cases where the habitat is significantly larger than the length-scale of dispersal, $\mu L/\pi A$ becomes not only a bound but also a close estimate for $p$.

We suspect that this general assertion about the sharpness of the symmetric bound could be proved, if formalized correctly---perhaps by generalizing the argument presented in Appendix A2 that $p$ approaches $\mu L/\pi A$ in a specific example.

In the case of asymmetric dispersal (Figure \ref{fig:plots}b), the bound $\mu L/2A$ remains an overestimate of simulated seed loss even as $A/L$ increases. The degree of overestimate depends on the bias direction of dispersal---for example, simulated seed loss is lower when dispersal is biased west, because fewer seeds are lost across the habitat's long north and south edges. It is interesting to note that although $\mu L/\pi A$ no longer serves as an upper bound in the asymmetric case, it does appear to capture the average seed loss among different dispersal directions as $A/L$ becomes large.

\section{Discussion}\label{sec:disc}

Whereas numerical simulation of $p$ requires detailed knowledge and coding of habitat geometry and plant dispersal kernels, the bounds of Theorem \ref{thm:main} require only mean dispersal distance $\mu$, habitat perimeter $L$, and habitat area $A$. We therefore anticipate that the bound will provide an accessible upper estimate and worst-case-scenario indicator of seed loss across a variety of applications. 
Because mean dispersal distance is the most challenging component to estimate in practice, we discuss methods for doing so in section \ref{subsec:estimating_mu} before outlining directions for future research in section \ref{subsec:future}.

\subsection{Estimating mean dispersal distance} \label{subsec:estimating_mu}

Estimating mean dispersal distance is perhaps most straightforward for wind-dispersed seeds, for which mechanistic models are well-developed. In particular, the ``ballistic" model
\begin{equation}\label{eq:muballistic}
    \mu=\frac{v_wh}{v_T}
\end{equation}
has long been used to estimate dispersal distance $\mu$ in terms of seed release height $h$, terminal velocity $v_T$, and wind speed $v_w$ (for a history, see \cite{nathan2012mechanistic}), and it remains the mean dispersal distance in the more sophisticated probabilistic WALD model that incorporates turbulent airflow \citep{katul2005mechanistic}.  Both wind speed $v_w$ and seed release height $h$ can be obtained from field measurements, while terminal velocity $v_T$ can be determined in a lab.  Sullivan and colleagues \citeyearpar{sullivan2018mechanistically} describe measurement methods for these quantities and report values for 50 grassland species at Cedar Creek Ecosystem Science Reserve in Isanti County, Minnesota. 

In the example in Section \ref{sec:examples}, estimating wind velocity was the most difficult part of using equation \eqref{eq:muballistic}---particularly matching wind data to the time and location of seed release. At the time of writing, the season of seed release was 5 months in the future and the height of wind sensors at the nearest weather stations was orders of magnitude higher than the height of seed release. 
To overcome such obstacles, one can use proxy measurements from similar sites (as we did) or wait for the season of seed release to take on-site field measurements.

For seeds dispersed by mechanisms other than wind, existing literature offers several starting points for estimating mean dispersal distance $\mu$. Rough estimates can be obtained from a meta-analysis of 148 mostly empirical sources conducted by Thomson and colleagues \citeyearpar{thomson2011seed}. In the course of studying relationships with plant height and seed mass, the authors condensed data on the mean dispersal distances of over 200 species into summary statistics grouped by 8 dispersal mechanisms (e.g. unassisted, water, vertebrate), presented in Appendix S3. Because the report does not provide data by species or list the primary research articles it draws upon, additional effort is required to narrow these order-of-magnitude bounds. 

First, one can consult existing literature for proxy estimates. For example, Bullock and colleagues published a synthesis report in \citeyear{bullock2017synthesis} that 
lists dispersal data sources by species for 144 vascular plant species in Table S4 of the online supplement. 
Second, novel analyses can be performed to estimate dispersal distances. Methods for estimating dispersal across a variety of mechanisms are reviewed by \cite{rogers2019total}. 
By considering a combination of dispersal estimates, one can develop a range that seems plausible for the species and site under consideration.

\subsection{Future directions}\label{subsec:future}

As noted in Section \ref{subsec:compare}, a salient area that remains to be explored is the conditions under which the quantities $\mu L/2A$ or $\mu L/\pi A$ serve not only as upper bounds but also as close estimates for $p$. If, as conjectured, the bounds are tight when dispersal is symmetric and mean dispersal distance is short relative to habitat size, then $\mu L/\pi A$ could be used to estimate seed loss directly in such settings.

One variable to consider while exploring tightness of bounds is the connectedness of habitat patches. Although our illustrations and example have featured a single connected habitat patch, the bounds also apply to a habitat ``patch" consisting of multiple nearby components that are disconnected by roadways or other land uses. (In this case, the total perimeter and area of all components are summed to determine $L$ and $A$, respectively.) Because dividing a habitat patch into two components via a narrow roadway may increase total perimeter $L$ without substantially altering area or seed loss, the bounds of Theorem \ref{thm:main} may significantly overestimate seed loss for disconnected habitats. (Of course, a roadway could also alter dispersal patterns by changing animal behavior, wind, or water flow in ways that violate the assumptions of our model, as discussed in the next paragraph.) 

Further theoretical efforts could work towards seed loss bounds that avoid one or several of the simplifying assumptions that underpin the dispersal model \eqref{eq:pdef}. 
First, we assume that the dispersal location random variable $\Xd$ and plant location random variable $\Xp$ are independent. However, dispersal patterns might vary according to plant location based on the effects of topographic slope and habitat edges on wind, water, or animal behavior. For example, in experimental patches of longleaf pine savanna embedded in a pine plantation matrix, Damschen and colleagues \citeyearpar{damschen2014fragmentation} found significant differences in wind velocities and propagule dispersal patterns at various locations within a patch. Patch shape influenced this phenomenon, with
long corridors of savanna in particular promoting uplifting and long distance dispersal.
Second, we assume that plant locations $\Xp$ follow a uniform distribution within a habitat, but variations in local site conditions could bias the plant distributions. 
In keeping with the orientation of this work towards estimates that use minimal information inputs, it would be particularly interesting to find bounds that allow for the variations described above without requiring their full characterization.
Lastly, we note that the binary classification of land as habitable or not habitable may be unrealistic for some ecosystems. Indeed, dispersal out of targeted reserves can have positive spillover effects on the surrounding landscape \citep{brudvig2009landscape}, pointing towards a more nuanced interpretation of seed loss from a habitat patch. Combining estimates of seed loss with estimates of establishment in the surrounding matrix (e.g. \cite{craig2011edge}) could yield insights into the potential benefits of dispersal out of a reserve.  

\backmatter

\subsection*{Statements and Declarations}

\bmhead{Competing Interests}
The authors declare no competing interests.

\bmhead{Data Availability}
Our estimate of \emph{A. syriaca}'s mean dispersal distance (see equation \eqref{eq:ballistic}) is based on plant height and seed terminal velocity data measured by \cite{sullivan2018mechanistically}, as well as hourly wind speed data from \cite{seeley2021hourly}. The satellite image of McKnight Prairie used in Figure \ref{fig:sim} is from \cite{google2024}.

\pagebreak

\begin{appendices}

\section{Sharpness of bounds}\label{secAsharp}

\subsection{Asymmetric dispersal}
In Theorem \ref{thm:main} we prove that, without any assumptions on the symmetry of the dispersal kernel, the proportion $p$ of seeds lost from a habitat patch with area $A$ and perimeter $L$ is bounded by
\begin{equation}\label{Aeq:p1}
    p\leq \frac{\mu L}{2A},
\end{equation}
where $\mu$ is the mean dispersal distance of the seeds. Rearranging, this is equivalent to
\begin{equation*}
    \frac{pA}{\mu L}\leq \frac{1}{2}.
\end{equation*}
Let $k=\dfrac{pA}{\mu L}$, so that the bound \eqref{Aeq:p1} is equivalent to $k\leq 1/2$. The following proposition implies that this bound is the lowest possible under the assumptions given in Section \ref{sec:modeling}.

\begin{proposition}
    There exist habitat geometries and dispersal kernels that make $k\equiv\dfrac{pA}{\mu L}$ arbitrarily close to $1/2$.
\end{proposition}

\begin{proof}
    Let $a<A$ and $b<B$ be positive real numbers.
    Consider a rectangular habitat
    \begin{equation*}
        H=\{(x,y)\in\mathbb{R}^2 \ | \ 0<x<A \text{ and } 0<y<B\},
    \end{equation*}
    and a dispersal kernel
    \begin{equation*}
        \fd(x,y)=\begin{cases}
            \frac{1}{ab} & 0 < x < a \text{ and } 0 < y < b \\
            0 & \text{otherwise}.
        \end{cases}
    \end{equation*}    
    We bound the factors in $k$ in terms of $a$, $b$, $A$, and/or $B$.
    First, we have the equality 
    \begin{equation}\label{Aeq:L}
        L=2(A+B).
    \end{equation} Second, from the definition \eqref{eq:mu} of mean dispersal distance we have 
    \begin{equation}\label{Aeq:mu}
    \begin{aligned}        \mu&=\iint\limits_{\mathbb{R}^2}|\xd|\fd(\xd)\, dA\\
    &=\int_0^b\int_0^a \sqrt{x^2+y^2}\frac{1}{ab}\,dx\,dy\\
    &\leq \int_0^b\int_0^a (x+y)\frac{1}{ab}\,dx\,dy\\
    &=\frac{a+b}{2}.
    \end{aligned}
    \end{equation}
    From \eqref{Aeq:L} and \eqref{Aeq:mu} it follows that
    \begin{equation}\label{Aeq:k1}
        k\geq \frac{pA}{\left(\frac{a+b}{2}\right)\cdot 2(A+B)}=\frac{pA}{(a+b)(A+B)}.
    \end{equation}
    Next we bound $pA$, which from Proposition \ref{prop:pHfd} is equivalent to the integral
    \begin{equation}\label{Aeq:pH}
        pA=\iint\limits_{\xp \in H} \ \ \iint\limits_{\x \in H^c} \fd(\x - \xp) \, dA \, dA.
    \end{equation}
    By restricting the domains of integration in equation \eqref{Aeq:pH}, we can only decrease the integral of the non-negative quantity $\fd$. Let $\xp=(x_p,y_p)$. Choosing bounds for $\x$ and $\xp$ over which $\fd(\x-\xp)= 1/ab$ (see Figure \ref{fig:A1}), we obtain  
    \begin{align}
        pA&\geq \int_0^B\int_{A-a}^{A}\left( \int_{y_p}^{y_p+b}\int_A^{x_p+a}\frac{1}{ab}\,dx\,dy\right)dx_p\,dy_p\label{eq:bounds}\\
        &=\int_0^B\int_{A-a}^{A}\left(
            \frac{x_p+a-A}{a}
            \right)dx_p\,dy_p\\
        &=\frac{B}{a} \int_{A-a}^{A}(x_p+a-A) \,dx_p,
    \end{align}
    which simplifies to
    \begin{equation}\label{Aeq:pHaB2}
        pA\geq\frac{aB}{2}.
    \end{equation}
    Combining inequalities \eqref{Aeq:k1} and \eqref{Aeq:pHaB2} yields
    \begin{equation*}
        k\geq\frac{aB}{2(a+b)(A+B)}.
    \end{equation*}
    For habitats with $B>>A$ the ratio $\dfrac{B}{A+B}$ can be made arbitrarily close to 1. Similarly, for dispersal kernels with $a>>b$, the ratio $\dfrac{a}{a+b}$ can also be made arbitrarily close to 1. It follows that $k$ can be made arbitrarily close to $\frac{1}{2}$. 
\end{proof}
\begin{figure}[t]
    \centering
    \includegraphics[width=0.4\textwidth]{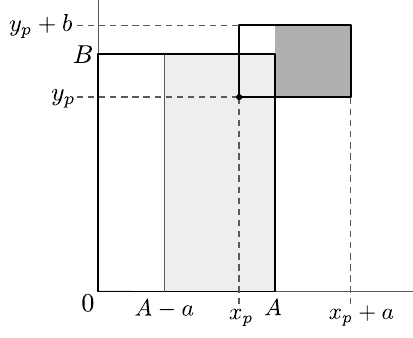}
    \vspace{1em}
    \caption{Quantities used to bound the integral \eqref{eq:bounds}. The light and dark grey rectangles represent the domains of integration for the outer and inner double integrals, respectively}
    \label{fig:A1}
\end{figure}
Thus, no values below $1/2$ suffice to bound $k$ from above. It is interesting to note that the ultimate quantities compared in this argument concern the aspect ratio of the rectangular habitat $H$ and the rectangular support for the dispersal kernel $\fd$. In particular, we imagine the aspect ratios becoming extreme in orthogonal directions. 

\pagebreak

\subsection{Symmetric dispersal}
Following the approach in Appendix A.1, we define 
\begin{equation}\label{eq:ksym}
    k=\frac{pA}{\mu L}
\end{equation}
so that the existence of habitat geometries and dispersal kernels that make $k$ arbitrarily close to $1/\pi$ implies that $p$ can be  arbitrarily close to its bound $\dfrac{\mu L}{\pi A}$ in Theorem \ref{thm:main} when dispersal is rotationally symmetric.\\

\begin{proposition}\label{prop:rotat_sharp}
    There exist habitat geometries and rotationally symmetric dispersal kernels that make $k\equiv\dfrac{pA}{\mu L}$ arbitrarily close to $1/\pi$.
\end{proposition}
\begin{proof}
Define a disk habitat $H$ of radius $R_H$; that is, 
\begin{equation*}
    H=\{\x: \ |\x|<R_H\},
\end{equation*}
and take the dispersal kernel $\fd$ to be a uniform distribution over a disk support with radius $R_d$:
\begin{equation*}
    \fd(\x)=\begin{cases}
        \dfrac{1}{\pi R_d^2} & \text{if }|\x|<R_d\\
        0 & \text{otherwise}.
    \end{cases}
\end{equation*}
Simple geometry gives that 
\begin{equation}\label{eq:L}
L=2\pi R_H.
\end{equation} And by definition we have
\begin{equation}\label{eq:mu_sym}
\begin{aligned}
    \mu&=\iint\limits_{\mathbb{R}^2} |\x|\fd(\x) \, dA\\
    &=\int_0^{2\pi}\int_0^{R_d} r \frac{1}{\pi R_d^2} r \, dr \, d\theta\\
    &=\frac{2}{3}R_d.
\end{aligned}
\end{equation}
Substituting for $pA$, $L$, and $\mu$ in equation \eqref{eq:ksym} using Proposition \ref{prop:pHfd} and equations \eqref{eq:L} and \eqref{eq:mu_sym} yields
\begin{equation*}
    k=\frac{3}{4\pi R_d R_H} \iint\limits_{\xp \in H} \ \ \iint\limits_{\x \in H^c} \fd(\x - \xp) \, dA \, dA.
\end{equation*}
Because the only plant positions $\xp$ that contribute to seed loss are those at distance between $R_H-R_d$ and $R_H$ from the origin, we express the outer double integral in polar coordinates as 
\begin{equation}
    k=\frac{3}{4\pi R_d R_H} \int_0^{2\pi} \int_{R_H-R_d}^{R_H}  \left( \ \ \iint\limits_{\x \in H^c} \fd(\x - \xp) \, dA \right)\, r_p \, dr_p \, d\theta
\end{equation}
where $\xp=(r_p,\theta_p)$.
We use Cartesian coordinates $\x=(x,y)$ for the inner double integral, with the $x$-axis aligned with the vector $\xp$ (see Figure \ref{fig:A2}). 
\begin{figure}[t]
    \centering
    \includegraphics[width=0.3\textwidth]{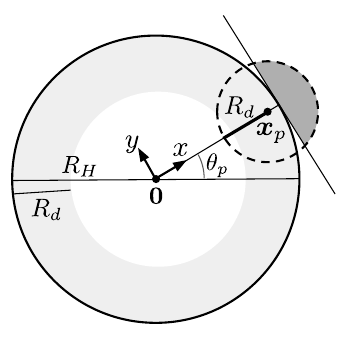}
    \caption{Quantities and coordinates used in the proof of Proposition \ref{prop:rotat_sharp}}
    \label{fig:A2}
\end{figure}
To simplify computations we restrict the domain of the inner integral to the segment of the circle of radius $R_d$ centered on $\xp$ that is cut  by the chord tangent to the habitat boundary at the $x$-axis.
Then
\begin{equation}
    k\geq\frac{3}{4\pi R_d R_H} \int_0^{2\pi} \int_{R_H-R_d}^{R_H}  \int_{R_H}^{r_p+R_d} \int_{-\sqrt{R_d^2-(x-r_p)^2}}^{\sqrt{R_d^2-(x-r_p)^2}} \ \frac{1}{\pi R_d^2} \ dy \, dx \, r_p \, dr_p \, d\theta.
\end{equation}
The $\theta$ and $y$ integrals can be evaluated directly, yielding
\begin{align}
    k&\geq\frac{3}{4\pi R_d R_H} 2\pi  \int_{R_H-R_d}^{R_H}  \int_{R_H}^{r_p+R_d} 2\sqrt{R_d^2-(x-r_p)^2}\frac{1}{\pi R_d^2}\, dx \, r_p \, dr_p\\
    &=\frac{3}{\pi R_d^3 R_H}\int_{R_H-R_d}^{R_H}  \int_{R_H}^{r_p+R_d} \sqrt{R_d^2-(x-r_p)^2} \ dx \, r_p \, dr_p.
\end{align}
Because $r_p>R_H-R_d$ over the domain of integration, one can eliminate the factor of $r_p$ from the integrand as follows:
\begin{equation}
    k\geq \frac{3(R_H-R_d)}{\pi R_d^3 R_H}\int_{R_H-R_d}^{R_H}  \int_{R_H}^{r_p+R_d} \sqrt{R_d^2-(x-r_p)^2} \ dx  \, dr_p.
\end{equation}
We first make the substitution $u=x-r_p$ and then reverse the order of integration:
\begin{align}
    k&\geq \frac{3(R_H-R_d)}{\pi R_d^3 R_H}\int_{R_H-R_d}^{R_H}  \int_{R_H-r_p}^{R_d} \sqrt{R_d^2-u^2} \ du  \, dr_p\\
    &=\frac{3(R_H-R_d)}{\pi R_d^3 R_H}\int_{0}^{R_d}  \int_{R_H-u}^{R_H} \sqrt{R_d^2-u^2} \ dr_p \, du\\
    &=\frac{3(R_H-R_d)}{\pi R_d^3 R_H}\int_{0}^{R_d} u \sqrt{R_d^2-u^2} \  du.
\end{align}
Evaluating the remaining integral yields $R_d^3/3$, and so 
\begin{equation}
    k\geq \frac{R_H-R_d}{\pi R_H}.
\end{equation}
By taking $R_H>>R_d$, $k$ can be made arbitrarily close to $1/\pi$, as desired.
\end{proof}

\section{Discussion of F}\label{secAF}

For the purpose of proving Theorem \ref{thm:main}, the vector field $\F$ was a convenient mathematical stepping stone tied to the idea of seed flow to build intuition. But a more precise physical interpretation is worth noting. Specifically, \textit{the flux of} $\F$ \textit{across a curve can be interpreted as the probability that a given seed would pass through that curve if all seeds traveled in straight lines}. Here, we show that this physical description of $\F$ leads logically to its mathematical definition in equation \eqref{eq:F}.

\begin{figure}[h!]
    \centering
    \includegraphics[width=0.4\textwidth]{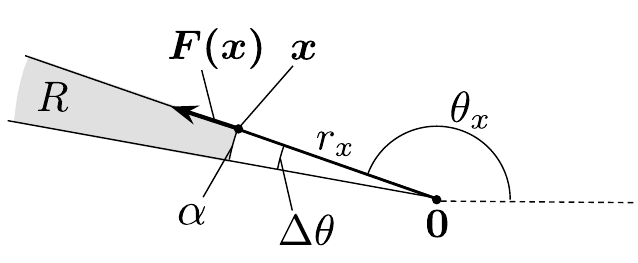}
    \caption{Notation used in the construction of the vector field $\F(\x)$}
    \label{fig:Fconstruct}
\end{figure}

Consider a single plant at the origin with dispersal kernel $\fd$. Given a location $\x$ in the landscape with polar coordinates $(r_x,\theta_x)$, we define an arc $\alpha$ starting at $\x$ and subtending an angle $\Delta\theta$. According to our interpretation above, the flux of $\F$ across $\alpha$ must equal the probability that a seed passes through $\alpha$. Since we are working under the hypothetical assumption that seeds travel in straight lines, the seeds that pass through $\alpha$ are exactly those that land in the infinite polar rectangle $R=\{(r,\theta) \ | \ r\geq r_x, \, \theta_x\leq \theta \leq \theta_x+\Delta\theta \}$ shown in Figure \ref{fig:Fconstruct}. Thus,
\begin{equation} \label{eq:fluxA1}
    \text{flux of $\F$ across $\alpha$} = P(\Xd\in R) = \int_{\theta_x}^{\theta_x+\Delta\theta}\int_{r_x}^\infty \fd(r,\theta) \, r \, dr \, d\theta.
\end{equation}
To be consistent with the straight-line dispersal hypothetical, we choose $\F$ to point radially outward from the origin. Then since $\F$ is normal to $\alpha$, its flux across $\alpha$ can also be expressed as
\begin{equation} \label{eq:fluxA2}
    \text{flux of $\F$ across $\alpha$} = \int_{\theta_x}^{\theta_x+\Delta\theta} |\F(r_x, \theta)| r_x \, d\theta,
\end{equation}
where $r_x \, d\theta$ serves as the differential arc length $ds$. After combining equations \eqref{eq:fluxA1} and \eqref{eq:fluxA2}, we may shed the $\theta$ integrals by dividing by $\Delta\theta$ and then taking the limit as $\Delta\theta$ approaches zero, leaving
\begin{equation}
    |\F(\x)|r_x = \int_{r_x}^\infty \fd(r,\theta_x) \, r \, dr.
\end{equation}
Solving for $|\F(\x)|$ and then multiplying by the unit vector $\x/r_x$ to recover $\F$'s direction gives our definition:
\begin{equation*}
    \F(\x) = \frac{\x}{r_x^2}\int_{r_x}^\infty \fd(r,\theta_x) \, r \, dr.
\end{equation*}

\end{appendices}

\bibliography{sn-bibliography}

@article{beckman2023causes,
  title={The causes and consequences of seed dispersal},
  author={Beckman, Noelle G and Sullivan, Lauren L},
  journal={Annual Review of Ecology, Evolution, and Systematics},
  volume={54},
  pages={403--427},
  year={2023},
  publisher={Annual Reviews}
}

@article{brudvig2009landscape,
  title={Landscape connectivity promotes plant biodiversity spillover into non-target habitats},
  author={Brudvig, Lars A and Damschen, Ellen I and Tewksbury, Joshua J and Haddad, Nick M and Levey, Douglas J},
  journal={Proceedings of the National Academy of Sciences},
  volume={106},
  number={23},
  pages={9328--9332},
  year={2009},
  publisher={National Acad Sciences}
}

@article{bullock2017synthesis,
  title={A synthesis of empirical plant dispersal kernels},
  author={Bullock, James M and Mallada Gonz{\'a}lez, Laura and Tamme, Riin and G{\"o}tzenberger, Lars and White, Steven M and P{\"a}rtel, Meelis and Hooftman, Danny AP},
  journal={Journal of Ecology},
  volume={105},
  number={1},
  pages={6--19},
  year={2017},
  publisher={Wiley Online Library}
}

@article{cheptou2008rapid,
  title={Rapid evolution of seed dispersal in an urban environment in the weed Crepis sancta},
  author={Cheptou, P-O and Carrue, O and Rouifed, S and Cantarel, A},
  journal={Proceedings of the National Academy of Sciences},
  volume={105},
  number={10},
  pages={3796--3799},
  year={2008},
  publisher={National Acad Sciences}
}

@misc{cooney2022effect, 
  author		= "Cooney, Mika and Hafner, Benjamin and Johnson, Shelby and Lee, Sean",
  title			= "The effect of habitat fragmentation on plant communities in a spatially implicit grassland model", 
  year			= "2022",
  note			= "In revision at Rose-Hulman Undergraduate Mathematics Journal"
}

@article{craig2011edge,
  title={Edge-mediated patterns of seed removal in experimentally connected and fragmented landscapes},
  author={Craig, Michael T and Orrock, John L and Brudvig, Lars A},
  journal={Landscape Ecology},
  volume={26},
  pages={1373--1381},
  year={2011},
  publisher={Springer}
}

@article{damschen2014fragmentation,
  title={How fragmentation and corridors affect wind dynamics and seed dispersal in open habitats},
  author={Damschen, Ellen I and Baker, Dirk V and Bohrer, Gil and Nathan, Ran and Orrock, John L and Turner, Jay R and Brudvig, Lars A and Haddad, Nick M and Levey, Douglas J and Tewksbury, Joshua J},
  journal={Proceedings of the National Academy of Sciences},
  volume={111},
  number={9},
  pages={3484--3489},
  year={2014},
  publisher={National Acad Sciences}
}

@misc{google2024,
  author = "Google\hspace{0.1cm}Maps",
  title = "Mc{K}night {P}rairie in {G}oodhue {C}ounty, {MN} [{S}atellite {M}ap]",
  year = "Accessed 21 March 2024",
  note = "\url{https://maps.app.goo.gl/TujkZDrtZUiHx32JA}"
}

@incollection{hanski2001population,
  author		= "Hanski, Ilkka",
  title			= "Population dynamic consequences of dispersal in local populations and in metapopulations",
  editor		= "Clobert, J and Danchin, E and Dhondt, AA and Nichols, JD",
  booktitle		= "Dispersal",
  pages			= "283--298",
  address		= "Oxford, UK",
  publisher		= "	Oxford University Press",
  year			= "2001"
}

@article{katul2005mechanistic,
  title={Mechanistic analytical models for long-distance seed dispersal by wind},
  author={Katul, GG and Porporato, A and Nathan, R and Siqueira, M and Soons, MB and Poggi, D and Horn, HS and Levin, SA},
  journal={The American Naturalist},
  volume={166},
  number={3},
  pages={368--381},
  year={2005},
  publisher={The University of Chicago Press}
}

@article{klausmeier1998extinction,
  title={Extinction in multispecies and spatially explicit models of habitat destruction},
  author={Klausmeier, Christopher A},
  journal={The American Naturalist},
  volume={152},
  number={2},
  pages={303--310},
  year={1998},
  publisher={The University of Chicago Press}
}

@article{levin2003ecology,
  title={The ecology and evolution of seed dispersal: a theoretical perspective},
  author={Levin, Simon A and Muller-Landau, Helene C and Nathan, Ran and Chave, J{\'e}r{\^o}me},
  journal={Annual Review of Ecology, Evolution, and Systematics},
  volume={34},
  number={1},
  pages={575--604},
  year={2003},
  publisher={Annual Reviews 4139 El Camino Way, PO Box 10139, Palo Alto, CA 94303-0139, USA}
}

@incollection{nathan2012dispersal,
  author		= "Nathan, Ran and Klein, Etienne and Robledo-Arnuncio, Juan J and Revilla, Eloy",
  title			= "Dispersal kernels: review",
  editor		= "Clobert, Jean and Baguette, Micheal and Benton, Tim G and Bullock, James M. ",
  booktitle		= "Dispersal Ecology and Evolution",
  pages			= "187--210",
  address		= "UK",
  publisher		= "Oxford University Press Oxford",
  year			= "2012"
}

@article{nathan2012mechanistic,
	title = {Mechanistic models of seed dispersal by wind},
	volume = {4},
	issn = {1874-1746},
	doi = {10.1007/s12080-011-0115-3},
	number = {2},
	journal = {Theoretical Ecology},
	author = {Nathan, Ran and Katul, Gabriel G. and Bohrer, Gil and Kuparinen, Anna and Soons, Merel B. and Thompson, Sally E. and Trakhtenbrot, Ana and Horn, Henry S.},
	year = {2011},
	pages = {113--132},
}

@article{rogers2019total,
  title={The total dispersal kernel: a review and future directions},
  author={Rogers, Haldre S and Beckman, Noelle G and Hartig, Florian and Johnson, Jeremy S and Pufal, Gesine and Shea, Katriona and Zurell, Damaris and Bullock, James M and Cantrell, Robert Stephen and Loiselle, Bette and others},
  journal={AoB Plants},
  volume={11},
  number={5},
  pages={plz042},
  year={2019},
  publisher={Oxford University Press US}
}

@misc{seeley2021hourly,
  author		= "Seeley, M.",
  title			= "Hourly climate data: Meterologic Measurements at {C}edar {C}reek Natural History Area ver 10.", 
  year			= "2021",
  note			= "\url{https://doi.org/10.6073/pasta/ce85da01f57b7639a06ab0dde4447b4c}"
}

@article{sullivan2018mechanistically,
    title={Mechanistically derived dispersal kernels explain species-level patterns of recruitment and succession},
    issue = {99},
    number = {11},
    journal={Ecology},
    author={Sullivan, Lauren L and Clark, Adam T and Tilman, David and Shaw, Allison K},
    year={2018},
    pages={2415--2420}
}

@article{thomson2011seed,
	title = {Seed dispersal distance is more strongly correlated with plant height than with seed mass},
	volume = {99},
	doi = {10.1111/j.1365-2745.2011.01867.x},
	number = {6},
	journal = {Journal of Ecology},
	author = {Thomson, Fiona J. and Moles, Angela T. and Auld, Tony D. and Kingsford, Richard T.},
	year = {2011},
	note = {\_eprint: https://onlinelibrary.wiley.com/doi/pdf/10.1111/j.1365-2745.2011.01867.x},
	pages = {1299--1307},
}

@article{tilman1994habitat,
  title={Habitat destruction and the extinction debt},
  author={Tilman, David and May, Robert M and Lehman, Clarence L and Nowak, Martin A},
  journal={Nature},
  volume={371},
  number={6492},
  pages={65--66},
  year={1994},
  publisher={Nature Publishing Group}
}


\end{document}